\documentclass[a4paper]{article}

\usepackage{amsthm,amsmath,amsfonts,amssymb}
\usepackage{tikz,a4wide}
\usepackage{xfrac}

\usepackage{mathtools}
\usepackage{enumerate}

\def\syz{\operatorname{Syz}}
\def\Span{\operatorname{span}}

\def\lt{\operatorname{lt}}

\newcommand\restr[2]{{
  \left.\kern-\nulldelimiterspace 
  #1 
  \right|_{#2} 
  }}

\newtheorem{theorem}{Theorem}
\newtheorem*{algorithm*}{Algorithm}

\newtheorem{thm}[theorem]{Theorem}
\newtheorem{corollary}[theorem]{Corollary}
\newtheorem{lemma}[theorem]{Lemma}

\newtheorem{algorithm}[theorem]{Algorithm}
\newtheorem{problem}[theorem]{Problem}

\newtheorem{example}[theorem]{Example}

\makeatletter
\newcommand{\myitem}[1]{%
\item[(#1)]\protected@edef\@currentlabel{#1}%
}
\makeatother

\def\eatspace#1{#1}
\def\step#1#2{\par\kern1pt\hangindent#2em\hangafter=1\noindent\rlap{\small#1}\kern#2em\relax\eatspace}

\let\set\mathbb
\def\<#1>{\langle#1\rangle}

\def\lt{\operatorname{lt}}

\begin{document}

\title{On the Problem of Separating Variables\\ in Multivariate Polynomial Ideals}

\author{Manfred Buchacher\footnote{M.\ Buchacher was supported by LIT-2022-11-YOU-214.}\ and
  Manuel Kauers\footnote{M.\ Kauers was supported by the Austrian FWF grants PAT 9952223 and I6130-N.}}
\date{%
  Institute for Algebra\\
  Johannes Kepler University\\
  A4040 Linz, Austria
}

\maketitle

\begin{abstract}
  For a given ideal $I\subseteq\set K[x_1,\dots,x_n,y_1,\dots,y_m]$
  in a polynomial ring with $n+m$ variables, we want to find all
  elements that can be written as $f-g$ for some $f\in\set K[x_1,\dots,x_n]$
  and some $g\in\set K[y_1,\dots,y_m]$, i.e., all elements of $I$
  that contain no term involving at the same time one of the $x_1,\dots,x_n$
  and one of the $y_1,\dots,y_m$.
  For principal ideals and for ideals of dimension zero, we give a algorithms
  that compute all these polynomials in a finite number of steps.
\end{abstract}

\section{Introduction}

The problem under consideration is as follows.
Given an ideal $I$ of a polynomial ring
\[\set K[x_1,\dots,x_n,y_1,\dots,y_m],\]
we want to know all elements of $I$ that can be written in the form $f-g$
for some $f\in\set K[x_1,\dots,x_n]$ and some $g\in\set K[y_1,\dots,y_m]$.
Such a polynomial $f-g$ is called \emph{separated} because it contains no
monomials that involve at the same time one of the $x_1,\dots,x_n$ and one
of the $y_1,\dots,y_m$.

It is not hard to see that the pairs $(f,g)\in\set K[x_1,\dots,x_n]\times\set K[y_1,\dots,y_m]$
such that $f-g$ is a separated element of an ideal $I$ of $\set K[x_1,\dots,x_n,y_1,\dots,y_m]$ form
a unital $\set K$-algebra with component-wise addition and multiplication. Indeed, $(1,1)$ is
clearly an element, and if $(f,g)$, $(f',g')$ are elements, then so are $(f+f',g+g')$ and
$(ff',gg')$, the latter because $ff'-gg'=(f-g)f'+g(f'-g')\in I$. We denote the set of all these
pairs $(f,g)$ by $A(I)$ and call it the algebra of separated elements of~$I$. Given a basis of
the ideal~$I$, we want to compute a set of generators of~$A(I)$.

Equations with separated variables have been studied at least since the
1950s~\cite{ehrenfeucht1956kryterium,davenport61,fried69,fried73,cassels69,bilu99,bilu00,cassou99}.
Early authors studied the algebraic curves defined by polynomials of the form $f(x)-g(y)$,
and in particular the question under which circumstances such a polynomial is irreducible,
and the structure of the corresponding function fields. 
Later, other aspects of the problem entered into the focus, for instance the
problem of finding integer roots of polynomials with separated variables~\cite{bilu2000diophantine}
or the relation of the separation problem to the problem of decomposing
polynomials~\cite{barton85,schicho95,binder96,aichinger11}.

The problem of finding separated polynomials in polynomial ideals has various applications.
One application is the intersection of $\set K$-algebras. For example,
computing
\[
\set K[u_1,\dots,u_n]\cap\set K[v_1,\dots,v_m]
\]
for given polynomials
\[
u_1,\dots,u_n,v_1,\dots,v_m\in\set K[t_1,\dots,t_k]
\]
is equivalent to finding all the separated polynomials $f-g$ in the ideal
\[
\<x_1-u_1,\dots,x_n-u_n,y_1-v_1,\dots,y_m-v_m>\cap\set K[x_1,\dots,x_n,y_1,\dots,y_m].
\]
Our own motivation comes from a different direction. In a study of generating
functions for lattice walk enumeration, Bousquet-Melou~\cite{melou16} finds the solution
of a certain functional equation using an interesting elimination technique.
She has certain power series $u_1,\dots,u_n$ in $\set K[z][[t]]$ and certain
power series $v_1,\dots,v_m$ in $\set K[z^{-1}][[t]]$ and needs to combine them to a
series that is free of~$z$. To do so, she finds polynomials $f$ and $g$ such that
$f(u_1,\dots,u_n)=g(v_1,\dots,v_m)$, and concludes that both sides of this
equation belong to $\set K[z][[t]]\cap\set K[z^{-1}][[t]]=\set K[[t]]$. We see
the development of algorithmic tools for finding separated polynomials as a key
step in turning Bousquet-Melou's technique into a general algorithm for solving
functional equations.

For ideals $I$ of a bivariate polynomial ring $\set K[x,y]$, the problem is well
understood. An algorithm for computing generators $I \cap (\mathbb{K}[x] + \mathbb{K}[y])$
was presented in~\cite{buchacher20}. Let us briefly sketch how this algorithm works.

Since every ideal $I\subseteq\mathbb{K}[x,y]$ is the intersection $I_0\cap I_1$ of a $0$-dimensional ideal $I_0$ and a principal ideal $I_1$, and because $A(I_0\cap I_1) = A(I_0)\cap A(I_1)$, it is sufficient to solve the problem for such ideals, and to be able to intersect the corresponding algebras. The algebra of separated polynomials of $I_0$ can be determined by first computing generators $p$, $q$ of its elimination ideals. The elements of $\mathbb{K}[x]\cdot p + \mathbb{K}[y]\cdot q$ are clearly separated, however, they do not necessarily make up all of $I_0\cap\left( \mathbb{K}[x] + \mathbb{K}[y]\right)$. For finding the remaining ones it is sufficient to make an ansatz whose degrees are bounded by the degrees of $p$ and $q$, reducing it, and solving a system of linear equations.

If $I_1$ is generated by some $p\in\mathbb{K}[x,y]\setminus (\mathbb{K}[x] \cup \mathbb{K}[y])$, then $A(I_1)$ is simple, i.e. generated by single element. Its generator corresponds to a separated polynomial $f-g$ that divides every other separated multiple of~$p$. To determine $f-g$, it is sufficient to know the degrees of $f$ and $g$. Finding $f-g$ then reduces to linear algebra.
It turns out that
there is always a grading on $\mathbb{K}[x,y]$ such that $\lt(f) - \lt(g)$ is the minimal separated multiple of the corresponding highest homogeneous component of~$p$. The problem of finding a degree bound for the minimal separated multiple of $p$ is thereby reduced to computing a separated multiple of a homogeneous polynomial.
It can be shown that a homogeneous bivariate polynomial has a separated multiple if and only if it is, possibly up to a rescaling of the variables, a product of pairwise distinct cyclotomic polynomials.
This can be checked by inspecting its roots.

Finally, the computation of the intersection of $A(I_0)$ and $A(I_1)$ is based on the fact that $A(I_0)$ has finite co-dimension as a $\set K$-linear subspace of $\mathbb{K}[x]\times\mathbb{K}[y]$ and that $A(I_1)$ is generated by a single element of $\mathbb{K}[x]\times\mathbb{K}[y]$. Any element of $A(I_0\cap I_1)$ is therefore a polynomial in the generator of $A(I_1)$, and (all) such polynomials can be found by (repeatedly) making an ansatz and solving a system of linear equations.

The present paper is about the separation problem for ideals $I$ of $\mathbb{K}[x_1,\dots,x_n,y_1,\dots,y_m]$ for arbitrary $n$ and~$m$. Our main result (Thm.~\ref{thm:5} and Corollary~\ref{cor:simpleness} below) is a constructive proof that shows that $A(I)$ is simple when $I$ is a principal ideal generated by an element of $\mathbb{K}[X,Y]\setminus\left( \mathbb{K}[X]\cup \mathbb{K}[Y]\right)$. We show that the computation of its generator can be reduced to the bivariate problem. This generalizes the corresponding result from~\cite{buchacher20}. Observing that the case of $0$-dimensional ideals can be treated in the same way as for bivariate polynomials, this then implies that we can proceed as in \cite{buchacher20} to compute a finite set of generators for $A(I)$ whenever $I$ is the intersection of a principal ideal and an ideal of dimension zero (Sect.~\ref{sec:dim0}). This implies in particular that $A(I)$ is finitely generated for such ideals.

However, in general $A(I)$ is not finitely generated, as
shown in Example~5.1 of~\cite{buchacher20}. This indicates that an extension of the techniques from the case $n=m=1$ to the case of arbitrary $n$ and $m$ is not straightforward, because there cannot be an algorithm
that computes for every given ideal a complete list of generators of~$A(I)$ in a finite number of steps.
In Sect.~\ref{sec:general}, we propose two procedures for enumerating generators of~$A(I)$.
We do not know if there is a procedure that terminates whenever $A(I)$ is finitely generated.

Throughout the paper, $\set K$ denotes a computable field of characteristic zero.
It is assumed that there is a way to check for a given element of $\set K$ whether
it is a root of unity. This is a fair assumption when $\set K$ is a number field
or a rational function field over a number field.
We write $X$ for $x_1,\dots,x_n$ and $Y$ for $y_1,\dots,y_m$ and consider the polynomial
ring $\set K[X,Y]$ in $n+m$ variables. When $p$ is a polynomial in the variable~$v$, we
denote the coefficient of $v^k$ in $p$ by $[v^k]p$ for any $k\in\mathbb{N}$.

\section{Principal Ideals}\label{sec:principal}

Consider a principal ideal $I=\<p>\subseteq\set K[X,Y]$.
If the generator belongs to $\set K[X]$ or to~$\set K[Y]$, then the
separation problem is not interesting. Let us exclude this case
and assume that $p\in\set K[X,Y]\setminus(\set K[X]\cup\set K[Y])$.
Our goal is to obtain information about $A(I)$ using our understanding
of the case $n=m=1$. Consider the ring homomorphism
\[
  \phi\colon\set K[X,Y]\to\set K(X,Y)[s,t]
\]
which maps each $x_i$ to $sx_i$ and each $y_j$ to~$ty_j$.
The codomain is a bivariate polynomial ring. Therefore, if $P=\phi(p)$ and $\bar I$ is the ideal generated by $P$
in $\set K(X,Y)[s,t]$, we know that the algebra $A(\bar I)$ is simple,
and we can compute a generator $(F,G)\in\set K(X,Y)[s]\times\set K(X,Y)[t]$.
If $A(\bar I)$ is trivial, then so is $A(I)$, because $\phi$ maps any
nontrivial element of $A(I)$ to a nontrivial element of $A(\bar I)$.
Suppose now that $A(\bar I)$ is nontrivial, and let $(F,G)$ be a generator.
As every nonzero $\set K(X,Y)$-multiple of a generator is again a generator,
we may assume that $(F,G)$ is such that $F$ and $G$ have no denominators
and that $F-G$ has no factor in $\set K[X,Y]$. Moreover, if $(F,G)$ is a
generator, then so is $(F+u,G+u)$ for every $u\in\set K(X,Y)$, because
$(1,1)$ is an element of the algebra. We may therefore further assume that
$(F,G)$ is such that $[s^0]F=0$.
We can alternatively assume that $[t^0]G=0$, but we cannot in general
assume that $[s^0]F$ and $[t^0]G$ both are zero. However, we can achieve
this situation by a change of variables, and it will be convenient to
do so. The following lemma provides the justification.

\begin{lemma}\label{lemma:0}
  Let $Q\in\set K^{n+m}$, and
  let $h\colon\set K[X,Y]\to\set K[X,Y]$ be the translation by~$Q$.
  Then $h$ induces an isomorphism of $\mathbb{K}$-algebras between $A(I)$ and $h(A(I))$. In particular, $h(A(I))=A(h(I))$ and a set of generators of $A(h(I)))$ can be obtained from a set of generators of $A(I)$ by applying $h$ to both components of each
  generator.
\end{lemma}
\begin{proof}
  Observe that $h$ maps $\set K[X]$ to $\set K[X]$ and $\set K[Y]$ to $\set K[Y]$,
  and that $h$ is invertible. Therefore,
  \[
  (f,g)\in A(I)\quad\Longleftrightarrow\quad(h(f),h(g))\in A(h(I))
  \]
  for all $f\in\set K[X]$ and all $g\in\set K[Y]$. The claim follows.
\end{proof}

If $Q$ is a point on which $p$ vanishes, then $h(p)$ is a polynomial with no
constant term. According to the lemma, it suffices to compute $A(\<h(p)>)$,
so we may assume without loss of generality that $p(0)=0$. We will then
also have $P(0)=0$, and then $(F-G)|_{s=0,t=0}=0$, so $[s^0]F=[t^0]G$, as
desired. If $\set K$ is not algebraically closed, a point $Q\in\mathbb{K}^{n+m}$ for which $p(Q) = 0$ may not exist. We may have to replace $\set K$
by some algebraic extension $\set K(\alpha)$ in order to ensure the existence
of a suitable~$Q$. By the following lemma, such algebraic extensions of the
coefficient field are harmless.

\begin{lemma}\label{lemma:alpha}
  Let $I\subseteq\set K[X,Y]$, let $\alpha$ be algebraic over~$\set K$,
  and let $J\subseteq\set K(\alpha)[X,Y]$ be the ideal generated by $I$ in $\set K(\alpha)[X,Y]$.
  If $A(J)$ is generated by a single element as $\set K(\alpha)$-algebra,
  then it has a generator with coefficients in~$\set K$, and this generator
  also generates $A(I)$ as $\set K$-algebra.
\end{lemma}
\begin{proof}
  Let $p_1,\dots,p_\ell\in\set K[X,Y]\subseteq\set K(\alpha)[X,Y]$ be ideal
  generators of~$I$ and consider a generator $(f,g)$ of~$A(J)$.
  We may assume that $(f,g)$ is not a $\set K(\alpha)$-multiple of $(1,1)$, because otherwise $A(J)$
  is trivial and there is nothing to show.

  There are $q_1,\dots,q_\ell\in\set K(\alpha)[X,Y]$ such that
  \[
  f-g=q_1p_1+\cdots+q_\ell p_\ell.
  \]
  If $\alpha$ is of degree~$d$, then $1,\alpha,\dots,\alpha^{d-1}$ is a
  $\set K$-vector space basis of~$\set K(\alpha)$.
  Write $f-g=\sum_{i=0}^{d-1}(f_i-g_i)\alpha^i$ for certain $f_i\in\set K[X]$ and $g_i\in\set K[Y]$,
  and write $q_j=\sum_{i=0}^{d-1}q_{i,j}\alpha^i$ for certain $q_{i,j}\in\set K[X,Y]$, so that
  \[
   \sum_{i=0}^{d-1}(f_i-g_i)\alpha^i=\sum_{i=0}^{d-1}(q_{i,1}p_1+\cdots+q_{i,\ell}p_\ell)\alpha^i.
  \]
  Since $p_1,\dots,p_\ell$ are free of~$\alpha$, we can compare coefficients and find that
  $(f_i,g_i)\in A(J)$. As $(f,g)$ is an algebra generator, each $(f_i,g_i)$ can be expressed
  as a polynomial of $(f,g)$ with coefficients in~$\set K(\alpha)$.
  As the degrees of nontrivial powers of $(f,g)$ exceed those of $(f,g)$, and therefore also
  those of $(f_i,g_i)$, we have in fact $(f_i,g_i)=u_i(f,g)+v_i(1,1)$ for certain
  $u_i,v_i\in\set K(\alpha)$. Since $(f,g)$ is not a $\set K(\alpha)$-multiple of $(1,1)$,
  at least one $(f_i,g_i)$ is not a $\set K(\alpha)$-multiple of $(1,1)$, and
  we can write $(f,g)$ as a $\set K(\alpha)$-linear combination of $(1,1)$ and this $(f_i,g_i)$.
  Then $(f_i,g_i)$ is a generator of $A(J)$ with coefficients in~$\set K$.

  By $f_i-g_i=q_{i,1}p_1+\cdots+q_{i,\ell}p_\ell$, we have $(f_i,g_i)\in A(I)$.
  Together with $A(I)\subseteq A(J)$, this implies that $(f_i,g_i)$ is also a generator of~$A(I)$.
\end{proof}

Assuming that $F,G$ are such that $[s^0]F=[t^0]G=0$,
the question is now what a generator $(F,G)$ of $A(\bar I)$ implies
about $A(I)$. Our answer to this question is Theorem~\ref{thm:5}, which
says that if $A(I)$ is nontrivial, then a generator of $A(I)$ can be
obtained from $(F,G)$.
In preparation for the proof of this theorem, we need a few lemmas.

\begin{lemma}\label{lemma:2}
  Let $F\in\set K[X,Y][s]$ be such that $[s^0]F=0$.
  Let the polynomials $u_0,\dots,u_k\in\set K[X,Y]$ be such that
  $u_0+u_1F + \cdots + u_kF^k$ has a factor $p$ in $\set K[X,Y]$.
  Suppose that $p$ is not a common factor of $u_0,\dots,u_k$.
  Then $p\mid F$.
\end{lemma}
\begin{proof}
  Without loss of generality, we may assume that $p$ is irreducible.
  (If it isn't, replace $p$ by one of its irreducible factors.)
  We show that the assumption $p\nmid F$ implies that $p$ is a
  common factor of $u_0,\dots,u_k$. Because of $[s^0]F=0$,
  the image of $F$ in $(\set K[X,Y]/\<p>)[s]$ is a polynomial of positive degree.
  Therefore, the images of $1,F,\dots,F^k$ in $(\set K[X,Y]/\<p>)[s]$ are
  linearly independent over $\set K[X,Y]/\<p>$.
  As the image of $u_0+u_1F + \cdots + u_kF^k$ in $(\set K[X,Y]/\<p>)[s]$
  is assumed to be zero, the images of $u_0,\dots,u_k$ must be zero,
  which means $p\mid u_i$ for all~$i$, as promised.
\end{proof}

\begin{lemma}\label{lemma:2a}
  Let $(F,G)\in\set K[X,Y][s]\times\set K[X,Y][t]$ be such that $[s^0]F=[t^0]G=0$.
  Suppose that $F-G$ has no factor in $\set K[X,Y]$.
  Let $u_0,\dots,u_k\in\set K(X,Y)$ be such that
  \[
  u_0\binom11+u_1\binom FG + \cdots + u_k\binom{F^k}{G^k}
  \in\set K[X,Y][s]\times\set K[X,Y][t].
  \]
  Then $u_0,\dots,u_k$ are in fact in $\set K[X,Y]$.
\end{lemma}
\begin{proof}
  Suppose otherwise and let $d\in\set K[X,Y]$ be the least common denominator
  of $u_0,\dots,u_k$ and $p$ be an irreducible factor of~$d$.
  Then
  \[
    p\mid du_0+du_1F+\cdots+du_kF^k
  \]
  and
  \[
    p\mid du_0+du_1G+\cdots+du_kG^k
  \]
  and $p\nmid du_i$ for at least one~$i$.
  By the previous lemma, this implies $p\mid F$ and $p\mid G$.
  But then $p\mid F-G$, in contradiction to the assumption.
\end{proof}

\begin{lemma}\label{lemma:1}
  Let $F\in\set K[X,Y][s]$ be such that $[s^0]F=0$.
  Let $k$ be a positive integer.
  Suppose that $[s^i]F^k$ is in $\set K[X]$ for every $i>(k-1)\deg_sF$.
  Then $F\in\set K[X][s]$.
\end{lemma}
\begin{proof}
  Write $F=c_1s+\cdots+c_ds^d$ with $d=\deg_sF$ and $c_1,\dots,c_d\in\set K[X,Y]$.
  We have $[s^{dk}]F^k=c_d^k$, which can only be in $\set K[X]$ if $c_d$ is.
  For $i=1,\dots,d-1$, the coefficient of $s^{dk-i}$ in $F^k$ is
  \[
    k c_d^{k-1}c_{d-i}+p(c_{d-i+1},c_{d-i+2},\dots,c_d)
  \]
  for a certain polynomial~$p$. This follows from the multinomial theorem.
  By induction on~$i$, it implies that also $c_1,c_2,\dots,c_{d-1}$ belong
  to $\set K[X]$, as claimed.
\end{proof}

\begin{lemma}\label{lemma:1a}
  Let $(F,G)\in\set K[X,Y][s]\times\set K[X,Y][t]$ be such that $[s^0]F=[t^0]G=0$.
  Suppose that $F-G$ has no factor in $\set K[X,Y]$.
  Let $u_0,\dots,u_k\in\set K[X,Y]$ be such that
  \[
    u_0\binom11+u_1\binom FG+\cdots+u_k\binom{F^k}{G^k}\in\set K[X][s]\times\set K[Y][t].
  \]
  Then $F\in\set K[X][s]$, $G\in\set K[Y][s]$, and $u_0,\dots,u_k\in\set K$.
\end{lemma}
\begin{proof}
  The $\deg_sF$ highest order terms of $F^k$ (w.r.t.~$s$) exceed the highest order terms of the lower powers of~$F$. (Note that the $u_0,\dots,u_k$ do not contain~$s$.)
  Since $u_0+u_1F+\cdots+u_kF^k$ belongs to $\set K[X][s]$ by assumption,
  neither $u_k$ nor the coefficients of the $\deg_sF$ highest order terms of $F^k$ can contain~$Y$.
  Therefore, by Lemma~\ref{lemma:1}, $F$ belongs to $\set K[X][s]$.

  As the $s$-degrees of the powers of $F$ are pairwise distinct, it follows furthermore
  that none of the $u_0,\dots,u_k$ can contain any~$Y$.

  By the same reasoning, we get that $G$ belongs to $\set K[Y][t]$ and that none of
  the $u_0,\dots,u_k$ can contain any~$X$, so in fact, we have $u_0,\dots,u_k\in\set K$.
\end{proof}

\begin{thm}\label{thm:5}
  Let $p\in\set K[X,Y]\setminus(\set K[X]\cup\set K[Y])$ be such that $p(0)=0$.
  Let $I=\<p>$, $P=\phi(p)$, and $\bar I=\<P>\subseteq\set K(X,Y)[s,t]$.
  Suppose that $A(\bar I)$ is not trivial and let $(F,G)\in\set K(X,Y)[s]\times\set K(X,Y)[t]$ be a generator
  such that $F$ and $G$ have no denominator, $F-G$ has no factor in $\set K[X,Y]$,
  and $F|_{s=0}=G|_{t=0}=0$.
  Then $A(I)$ is nontrivial if and only if $F\in\set K[X][s]$ and $G\in\set K[Y][t]$
  and $F|_{s=1}\neq G|_{t=1}$.
  In this case, $(F|_{s=1},G|_{t=1})$ is a generator of $A(I)$.
\end{thm}
\begin{proof}
``$\Leftarrow$'':
If $F$ and $G$ are as in the assumption, then $F-G$ is a $\set K(X,Y)[s,t]$-multiple of~$P$, say
$F-G=QP$ for some $Q\in\set K(X,Y)[s,t]$.
Since $P$ has no factor in $\set K[X,Y]$ and $F-G$ has no denominator, it follows that $Q$ has no denominator.
Therefore, setting $s=1$ and $t=1$ shows that $F|_{s=1}-G|_{t=1}$ is a separated multiple of $p$ and
therefore an element of~$I$. It follows that $A(I)$ contains $(F|_{s=1},G|_{t=1})$. It remains to
show that this is not a $\set K$-multiple of $(1,1)$. If it were, then $F|_{s=1}-G|_{t=1}=0$,
which is excluded by assumption on $F$ and~$G$.

 ``$\Rightarrow$'': If $A(I)$ is nontrivial,
it contains some pair $(f,g)\in\set K[X]\times\set K[Y]$ that is not a $\set K$-multiple of $(1,1)$.
Then $(\phi(f),\phi(g))$ is a nontrivial element of $A(\bar I)$.
Then there are $u_0,\dots,u_k\in\set K(X,Y)$ such that
\[
\binom{\phi(f)}{\phi(g)}=u_0\binom11+u_1\binom FG+\cdots+u_k\binom{F^k}{G^k}.
\]
The left hand side has no denominator in $\set K[X,Y]$, because $f$ and $g$ are polynomials.
Therefore, by Lemma~\ref{lemma:2a}, $u_0,\dots,u_k$ belong to $\set K[X,Y]$.
Next, by Lemma~\ref{lemma:1a}, it follows that $F\in\set K[X][s]$, $G\in\set K[Y][t]$,
and $u_0,\dots,u_k\in\set K$.

It remains to show that $F|_{s=1}\neq G|_{t=1}$. If they were equal, then they would
be in~$\set K$, because $F|_{s=1}$ does not contain $Y$ and $G|_{t=1}$ does not contain~$X$.
Then $(F|_{s=1},G|_{t=1})$ would be a $\set K$-multiple of $(1,1)$, and
\[
u_0\binom11+u_1\binom{F|_{s=1}}{G|_{t=1}}+\cdots+u_k\binom{(F|_{s=1})^k}{(G|_{t=1})^k}
\]
would also be a $\set K$-multiple of $(1,1)$. This is impossible, because $(f,g)$ is assumed
not to be a $\set K$-multiple of $(1,1)$.

This completes the argument for the direction ``$\Rightarrow$''.
In this argument, we have shown that every element of $A(I)$ can be written as a polynomial
in $(F|_{s=1},G|_{t=1})$.
This construction also implies the additional claim about the generator of $A(I)$.
\end{proof}

\begin{example}
  \begin{enumerate}
  \item If $I$ is generated by $x_1^2+2x_1x_2+x_2^2+x_1y+x_2y+y^2$, then both $A(I)$ and $A(\bar I)$
    are nontrivial.
    They are generated by $((x_1+x_2)^3, y^3)$ and $((x_1+x_2)^3s^3,y^3t^3)$, respectively.
  \item If $I$ is generated by $x_1^2+x_1x_2+x_2^2+x_1y+x_2+y^2$, then $A(I)$ and $A(\bar I)$ both
    are trivial.
  \end{enumerate}
\end{example}

There is no example where $A(\bar I)$ is trivial but $A(I)$ is not, because $\phi$ maps nontrivial elements
of $A(I)$ to nontrivial elements of~$A(\bar I)$.
Conversely, we have also not found any example of a principal ideal $I$ where $A(I)$ is trivial but $A(\bar I)$ is not, and we
suspect that no such example exists. However, as we will see in Example~\ref{ex:final}, there are such examples when $I$ is not principal.

\begin{corollary}\label{cor:simpleness}
  For every $p\in\set K[X,Y]\setminus(\set K[X]\cup\set K[Y])$, the algebra $A(\<p>)$
  is simple.
\end{corollary}
\begin{proof}
  We argue that all assumptions in Thm.~\ref{thm:5} are ``without loss of generality.''
  First, by Lemmas~\ref{lemma:0} and \ref{lemma:alpha}, we can assume that $p(0)=0$.
  If $A(\<p>)$ is trivial, there is nothing to prove.
  If $A(\<p>)$ is not trivial, then so is $A(\<P>)$.
  If $(F,G)$ is any generator of $A(\<P>)$, then so is
  every $\alpha(F,G)+\beta(1,1)$ for any choice $\alpha\in\set K(X,Y)\setminus\{0\}$ and $\beta\in\set K(X,Y)$.
  By a suitable choice of $\alpha$ and $\beta$, we can meet the assumptions imposed
  on $(F,G)$ in Thm.~\ref{thm:5}.
  According to the theorem, then $(F|_{s=1},G|_{t=1})$ is a generator of~$A(I)$.
\end{proof}

The assumption that $p$ does not belong to $\set K[X]$ or to $\set K[Y]$ is necessary.
For example, if $p\in\set K[X]$, the algebra $A(I)$ consists of all $(f+c,c)$ where $f\in\mathbb{K}[X]\cdot p$ and $c\in\mathbb{K}$,
and while this is a concise description of~$A(I)$, such an algebra need not be finitely
generated. To see this, consider $p=x_1x_2\in\set K[x_1,x_2]$. The $x_2$-degree of any
nontrivial power of any nontrivial $\set K[x_1,x_2]$-multiple of $p$ will be at least~$2$,
so every element $x_1^kx_2$ of the algebra can only be a $\set K$-linear combination of
generators. Because of $\dim_{\set K}x_1x_2\set K[x_1]=\infty$, there must be infinitely
many generators.

We have just seen that the algebra $A(I)$ is simple whenever the ideal $I$ is generated by a polynomial $p$ of $\mathbb{K}[X,Y]$ that is not an element of $\mathbb{K}[X] \cup \mathbb{K}[Y]$. We now give a characterization of the generator of $A(I)$ in terms of certain divisibility relations. It is based on the following generalization of a theorem by Fried and MacRae~\cite{fried69}. For a proof we refer to~\cite{schicho95}. See also~\cite{aichinger11}.
\begin{theorem}\label{thm:fried}
Let $f,F\in\mathbb{K}[X]$ and $g,G\in\mathbb{K}[Y]$ be non-constant polynomials. The following are equivalent:
\begin{enumerate}
\item There exists $h\in\mathbb{K}[t]$ such that $F = h(f)$ and $G = h(g)$.
\item $f-g$ divides $F-G$ in $\mathbb{K}[X,Y]$.
\end{enumerate}
\end{theorem}
Let $F-G \in I \cap \left( \mathbb{K}[X] + \mathbb{K}[Y]\right)$ such that $(F,G)\in A(I)$. If $A(I)$ is simple and generated by $(f,g)\in\mathbb{K}[X]\times\mathbb{K}[Y]$, then $(F,G) = (h(f),h(g))$ for some $h\in\mathbb{K}[t]$. The previous theorem implies that $f-g$ divides $F-G$ in $\mathbb{K}[X,Y]$. As a consequence of Corollary~\ref{cor:simpleness} and Theorem~\ref{thm:fried} we therefore have the following.
\begin{corollary}
Let $p\in\mathbb{K}[X,Y]$. If $p$ has a separated multiple, then it has one that divides any other of its separated multiples.
\end{corollary}
If $p$ has a separated multiple and the corresponding algebra is generated by $(f,g)$, then $f-g$ is referred to as the \textit{minimal separated multiple} of $p$. It is unique up to a multiplicative constant.

\section{Ideals of Dimension Zero}\label{sec:dim0}

For ideals of dimension zero, the technique proposed in \cite{buchacher20} for the case $n=m=1$ generalizes
more or less literally to arbitrary $n$ and~$m$. We therefore give only an informal summary here
and refer to \cite{buchacher20} for a more formal discussion.

If $I\subseteq\set K[X,Y]$ has dimension zero, then it
contains a nonzero univariate polynomial for each of the variables. Denote these polynomials by
$p_1,\dots,p_n,q_1,\dots,q_m$. Being univariate, these polynomials are in particular separated.
This implies that $A(I)$ contains at least all pairs $(p,q)$ where $p$ is a
$\set K[X]$-linear combination of $p_1,\dots,p_n$ and $q$ is a $\set
K[Y]$-linear combination of $q_1,\dots,q_m$.  If $(f,g)$ is any other element
of~$A(I)$, we can add an arbitrary $\set K[X]$-linear combination of
$p_1,\dots,p_n$ to $f$ and an arbitrary $\set K[Y]$-linear combination of
$q_1,\dots,q_m$ to $g$ and obtain another element of~$A(I)$. It is therefore
enough to search for elements $(f,g)$ of $A(I)$ with $\deg_{x_i}f<\deg_{x_i}p_i$
and $\deg_{y_j}g<\deg_{y_j}q_j$ for all $i$ and~$j$.  This restricts the search
to a finite dimensional vector space. We can make an ansatz with undetermined
coefficients for $f$ and~$g$, compute its normal form with respect to a
Gr\"obner basis of~$I$, equate its coefficients to zero and solve the resulting
linear system for the unknown coefficients in~$\set K$. The solutions together with the
$p_1,\dots,p_n$ and their $X$-multiples as well as the $q_1,\dots,q_m$ and their
$Y$-multiples then form a set of generators of~$A(I)$.

\begin{example}
  Let $I\subseteq\set K[x_1,x_2,y_1,y_2]$ be the ideal generated by
  \begin{alignat*}1
    &x_1+x_2+y_1+y_2,\\
    &x_1 x_2+x_1 y_1+x_1 y_2+x_2 y_1+x_2 y_2+y_1 y_2,\\
    &x_1 x_2 y_1+x_1 x_2 y_2+x_1 y_1 y_2+x_2 y_1 y_2,\\
    &x_1 x_2 y_1 y_2-1.
  \end{alignat*}
  Its elimination ideals are
  \begin{alignat*}1
    I\cap\set K[x_1,x_2]&=\<x_1^3+x_1^2 x_2+x_1 x_2^2+x_2^3,x_2^4+1>,\\
    I\cap\set K[y_1,y_2]&=\<y_1^3+y_1^2 y_2+y_1 y_2^2+y_2^3,y_2^4+1>.
  \end{alignat*}
  Denoting the two generators of $I\cap\set K[x_1,x_2]$ by $p_1,p_2$, respectively,
  polynomial division shows that this ideal is generated as a $\set K$-algebra by $x_1^i x_2^j p_1$
  for $i=0,1,2$ and $j=0,1,2,3$ and $x_1^i x_2^j p_2$
  for $i=0,1,2,3$ and $j=0,1,2$. Similarly, we get a finite set of generators for the other elimination ideal.

  It remains to check whether $A(I)$ contains any elements $(p,q)$ where all terms
  in $p$ have $x_1$-degree less than $3$ and $x_2$-degree less than~$4$, and all
  terms in $q$ have $y_1$-degree less than $3$ and $y_2$-degree less than~$4$.
  It turns out that the following pairs form a basis of the $\set K$-vector space
  of all these elements:
  \begin{alignat*}1
    \binom{x_1^2+x_2^2}{-y_1^2-y_2^2},
    \binom{x_1+x_2}{-y_1-y_2},
    \binom{x_1x_2}{y_1^2+y_1 y_2+y_2^2},
    \binom{x_1^2 x_2+x_1 x_2^2}{-y_1^2y_2-y_1 y_2^2},
    \binom{x_1^2 x_2^2}{y_1^2 y_2^2}.
  \end{alignat*}
  These pairs together with the generators of the two elimination ideals form
  a finite set of generators of~$A(I)$.
\end{example}

As a $\set K$-linear subspace of $\set K[X]\times\set K[Y]$, the algebra $A(I)$
for an ideal $I$ of dimension zero has finite co-dimension. From the algebra
generators of $A(I)$ computed as described above, we can obtain a basis of a
vector space $V$ such that $V\oplus A(I)=\set K[X]\times\set K[Y]$, and for
every $(f,g)\in\set K[X]\times\set K[Y]$ we can compute a pair $(\tilde f,\tilde
g)\in V$ such that $(f,g) - (\tilde f, \tilde g) \in A(I)$. This
amounts to Lemma~2.4 of \cite{buchacher20}.

In the case $n=m=1$, every ideal can be written as the intersection of an ideal
of dimension zero and a principal ideal. This is no longer true in the general
case. However, if an ideal $I\subseteq\set K[X,Y]$ happens to be the intersection of
an ideal $I_0\subseteq\set K[X,Y]$ of dimension zero and a principal ideal
$I_1\subseteq\set K[X,Y]$, then we can continue as in Sect.~4 of \cite{buchacher20} and
obtain a finite set of generators for~$A(I)$.

Algorithm~4.3 of \cite{buchacher20} relies on $A(I_0\cap I_1)=A(I_0)\cap A(I_1)$ and uses that
$A(I_0)$ has finite codimension and $A(I_1)$ is generated by a single
element. It makes an ansatz for a polynomial in the generator of $A(I_1)$, then
finds an equivalent element in $V$ and forces that element
to zero. This results in a system of linear equations for the coefficients of the ansatz, whose solutions give rise to elements of
$A(I_0)\cap A(I_1)$. The search is repeated with an ansatz of larger and larger
degree, but always excluding all monomials that are $\set N$-linear
combinations of degrees of generators found earlier. Since $(\set N,+)$ is a
noetherian monoid, after finitely many repetitions there are no monomials left
and the list of generators is complete.

The correctness of this algorithm does not depend on the assumption $n=m=1$ but
extends literally to the case of arbitrary $n$ and~$m$. We can therefore record
the following corollary to Thm.~\ref{thm:5}.

\begin{corollary}
  Let $I\subseteq\set K[X,Y]$ be such that $I=I_0\cap I_1$ for some ideal $I_0$ of dimension zero
  and some principal ideal~$I_1$ whose generator is not in $\set K[X]\cup\set K[Y]$.
  Then $A(I)$ is finitely generated, and there is an algorithm for computing a finite set of
  generators.
\end{corollary}

\begin{example}
  As a minimalistic example, consider the
  ideal $I=I_0\cap I_1\subseteq\set K[x_1,x_2,y_1,y_2]$ with
  \[
  I_0=\<x_1-1,x_2-1,y_1-2,y_2-2>
  \quad\text{and}\quad
  I_1=\<x_1^2 + x_1 y_2 + y_2^2>.
  \]
  The algebra $A(I_0)$ is generated by $(x_1-1,0),(x_2-1,0),(0,y_1-2),(0,y_2-2)$,
  and the algebra $A(I_1)$ is generated by $g=(x_1^3,y_2^3)$.
  We need to find all univariate polynomials $p$ such that $p(g)\in A(I_0)$.

  Modulo the $\set K$-vector space $A(I_0)$, the element $g$ itself
  is equivalent to $(0,7)$, and the element $g^2$ is equivalent to
  $(0,63)$. Therefore, $g^2-9g$ is an element of~$A(I_0)$.
  This reduces the search to polynomials involving only odd powers of~$g$.
  As the element $g^3$ is equivalent modulo $A(I_0)$ to $(0,511)$,
  we find the additional element $g^3-73g$ of~$A(I)$.
  Since $2\set N+3\set N=\set N\setminus\{0,1\}$ and $A(I_0)$ does not
  contain any element of the form $\alpha g+\beta$, we can conclude that
  $A(I)=\set K[g^2-9g,g^3-73g]$.
\end{example}

\section{Arbitrary Ideals}\label{sec:general}

For an arbitrary ideal $I$ of $\set K[X,Y]$, the algebra of separated polynomials is
in general not finitely generated. It is therefore impossible to give an algorithm
that computes a complete basis in a finite number of steps. The best we
can hope for is a procedure that enumerates a set of generators and runs forever if $A(I)$ is not finitely generated, yet terminates if $A(I)$ is finitely generated. Unfortunately, we cannot offer such a procedure. However, if we drop the latter requirement, it is not hard to come up with an algorithmic solution.

For any fixed $d\in\set N$, we can find all $(f,g)\in A(I)$ where $f$ and
$g$ have total degree at most $d$ by linear algebra, similar as in the case
of zero dimensional ideals. Make an ansatz
\begin{alignat*}1
  f&=\sum_{e_1+\dots+e_n\leq d}\alpha_{e_1,\dots,e_n}x_1^{e_1}\cdots x_n^{e_n},\\
  g&=\sum_{e_1+\dots+e_m\leq d}\beta_{e_1,\dots,e_m}y_1^{e_1}\cdots y_m^{e_m}
\end{alignat*}
with undetermined coefficients $\alpha_{e_1,\dots,e_n},\beta_{e_1,\dots,e_n}$
and compute the normal form of $f-g$ with respect to a Gr\"obner basis
of~$I$. The result will be a polynomial in $X,Y$ whose coefficients are $\set
K$-linear combinations of the undetermined coefficients. Force these
coefficients to zero and solve the resulting linear system. The result
translates into a basis of the $\set K$-vector space of all pairs $(f,g)\in
A(I)$ with $f$ and $g$ of total degree at most~$d$.
By repeating this computation for $d=1,2,3,\dots$ indefinitely, we will get
a set of generators of~$A(I)$. In fact, these generators generate $A(I)$ not
only as $\set K$-algebra but even as $\set K$-vector space. This is more than
we want. We can eliminate some of the redundance in the output by discarding
from the ansatz all terms that are powers of leading terms of generators that
have been found in earlier iterations, but the approach nevertheless seems
brutal as the size of the linear system will grow rapidly with increasing~$d$.

An alternative procedure for enumerating algebra generators of $A(I)$ uses
Gr\"obner bases instead of linear algebra. For this procedure, we reuse the
idea of Sect.~\ref{sec:principal} and exploit the fact that we know how to compute
a (finite) set of generators of $A(\bar I)$ for every ideal $\bar I$ of
a bivariate polynomial ring.

Like in Sect.~\ref{sec:principal}, we consider the homomorphism
\[
 \phi\colon\set K[X,Y]\to\set K(X,Y)[s,t]
\]
which maps each $x_i$ to $sx_i$ and each $y_j$ to $ty_j$. Let $p_1,\dots,p_\ell\in\set K[X,Y]$
be generators of $I\subseteq\set K[X,Y]$, let $P_i=\phi(p_i)$ for $i=1,\dots,\ell$, and let
$\bar I$ be the ideal generated by $P_1,\dots,P_\ell$ in $\set K(X,Y)[s,t]$.
The algebra $A(\bar I)$ is finitely generated. Let $B_1,\dots,B_u$ be a choice of generators.
The homomorphism $\phi$ maps every element of $A(I)$ to an element of~$A(\bar I)$, and every
such element can be written as a polynomial in $B_1,\dots,B_u$ with coefficients in $\set K(X,Y)$.
Therefore, in order to find elements of~$A(I)$, we search for elements of $\set K(X,Y)[B_1,\dots,B_u]$
that become elements of $A(I)$ after setting $s$ and $t$ to~$1$.
This can be done effectively as soon as we can solve the following problem:

\begin{problem}\label{problem:18}
  Given: generators $p_1,\dots,p_\ell$ of $I$ and some elements
  $(F_1,G_1),\dots,(F_k,G_k)$ of $A(\bar I)$

  Find: a $\set K$-vector space basis of the set of all elements of $A(I)$ that can be
  obtained from a $\set K(X,Y)$-linear combination of $(F_1,G_1),\dots$, $(F_k,G_k)$ by
  setting $s$ and $t$ to~$1$.
\end{problem}

With an algorithm for solving this problem, we can get a procedure that enumerates generators of~$A(I)$.
For $d=1,2,\dots$ in turn, the procedure calls the algorithm with all monomials in $B_1,\dots,B_u$ of
degree at most $d$ as $(F_1,G_1),\dots,(F_k,G_k)$.

In the remainder of this section, we discuss an algorithm for solving Problem~\ref{problem:18}.
We first give a high-level description of the algorithm and prove that the approach is sound and
complete. Afterwards, we show that each of the steps can be effectively computed.

\begin{algorithm}\label{alg:19}
  Input/Output: as specified in Problem~\ref{problem:18}

  \medskip
  \step 11 Compute a basis of the $\set K[X,Y]$-module
  \[
    M := \Span_{\set K(X,Y)}(F_1-G_1,\dots,F_k-G_k)\cap\underbrace{\<\phi(p_1),\dots,\phi(p_\ell)>}_{\subseteq\set K[X,Y][s,t]}.
  \]
  Write the elements $F-G$ of $M$ in the form $(F,G)$,
  so that $M$ becomes a submodule of $\set K[X,Y][s]\times\set K[X,Y][t]$.
  (Include the pair $(1,1)$ among the generators.)
  \step 21 Compute bases of the $\set K[X]$-module
  \[
    M_X:=\{\,(F,G)\in M: F\in\set K[X][s]\,\}
  \]
  and the $\set K[Y]$-module
  \[
    M_Y:=\{\,(F,G)\in M: G\in\set K[Y][t]\,\}.
  \]
  \step 31 Compute a basis of the $\set K$-vector space $M_X\cap M_Y$.
  \step 41 Set $s=t=1$ in the basis elements and return the result.
\end{algorithm}

\begin{theorem}
  Alg.~\ref{alg:19} is sound and complete.
\end{theorem}
\begin{proof}
\textit{Soundness.} We show that every pair $(f,g)$ in the output indeed belongs to~$A(I)$.
If $(f,g)$ is an element of the output, then it is clear from Step~3
and the definition of $M_X,M_Y$ that $f\in \set K[X]$ and $g\in \set K[Y]$. We need to show
that $f-g\in I$.
Let $F,G$ be the polynomials from which $f$ and $g$ are obtained by setting $s$ and~$t$ to~$1$.
Then $(F,G)$ is an element of~$M$, therefore $F-G$ is an element of $\<\phi(p_1),\dots,\phi(p_l)>$,
and therefore $f-g$ is an element of~$I$.

\textit{Completeness.} We show that if $(f,g)\in A(I)$ is such that the corresponding
\[
(F,G)\in\set K(X,Y)[s]\times\set K(X,Y)[t]
\]
is a $\set K(X,Y)$-linear combination of the elements $(F_1,G_1),\dots,(F_k,G_k)$, then it is
a $\set K$-linear combination of the output pairs.
By assumption, $F-G\in\Span_{\set K(X,Y)}(F_1-G_1,\dots,F_k-G_k)$.
Also, since $f-g\in I$, we have $F-G\in\<\phi(p_1),\dots,\phi(p_l)>$.
Therefore, $(F,G)$ belongs to the module $M$ computed in Step~1.
Moreover, we have $F\in \set K[X][s]$ and $G\in \set K[Y][t]$ because $f\in \set K[X]$ and $g\in \set K[Y]$,
so $(F,G)\in M_X\cap M_Y$. The claim follows.
\end{proof}

Step~4 of Alg.~\ref{alg:19} is trivial, and Step~2 is a standard application of Gr\"obner
bases. For example, in order to get a basis of $M_X$, it suffices to compute a Gr\"obner
basis of $M$ with respect to a TOP term order that eliminates~$Y$, and to discard from it
all elements which have a $Y$ in the first component~\cite[Definition 3.5.2]{adams2022introduction}.
Steps~1 and~3 require more explanation.

For Step~1, we divide the problem into two substeps. First we compute a basis of the $\set K[X,Y]$-module
\[
 N:=\Span_{\set K(X,Y)}(F_1-G_1,\dots,F_k-G_k)\cap\set K[X,Y][s,t],
\]
and then we obtain a basis of $M$ by computing the intersection of this $N$ with the ideal generated
by $\phi(p_1),\dots,\phi(p_\ell)$ in~$\set K[X,Y][s,t]$. The two substeps are provided by the following lemmas.

\begin{lemma}\label{lemma:cap}
  For any given $q_1,\dots,q_k\in\set K[X,Y][s,t]$, we can compute a basis of the $\set K[X,Y]$-module
  \[
    \Span_{\set K(X,Y)}(q_1,\dots,q_k)\cap\set K[X,Y][s,t].
  \]
\end{lemma}
\begin{proof}
  As only finitely many monomials appear in $q_1,\dots,q_k$, we can view them as elements of a finitely
  generated $\set K[X,Y]$-submodule of $\set K[X,Y][s,t]$. We may identify this submodule with $\set K[X,Y]^n$
  for some~$n$. In this identification, $\Span_{\set K(X,Y)}(q_1,\dots,q_k)$ is a certain subspace
  of~$\set K(X,Y)^n$. Let $A\in\set K(X,Y)^{m\times n}$ be a matrix whose kernel is this subspace.
  Such a matrix exists and can be easily constructed by means of linear algebra.
  As multiplying $A$ by a nonzero element of $\set K(X,Y)$ does not change the kernel, we may assume
  that $A$ belongs to $\set K[X,Y]^{m\times n}$.
  Let $a_1,\dots,a_n\in\set K[X,Y]^m$ be its columns.
  Then
  \[
  \Span_{\set K(X,Y)}(q_1,\dots,q_k)\cap\set K[X,Y]^n
  =\syz(a_1,\dots,a_m).
  \]
  The computation of a basis of the syzygy module is a standard application of Gr\"obner bases.
\end{proof}

\begin{lemma}
  Let $N$ be a finitely generated $\set K[X,Y]$-submodule of $\set K[X,Y][s,t]$ and
  let $J$ be an ideal of $\set K[X,Y][s,t]$.
  Then $N\cap J$ is a finitely generated submodule of $\set K[X,Y][s,t]$, and we can
  compute a basis of it from a module basis of $N$ and an ideal basis of~$J$.
\end{lemma}
\begin{proof}
  Let $n_1,\dots,n_r$ be module generators of~$N$ and $p_1,\dots,p_k$ be ideal generators of~$J$.
  An element $q$ of $\set K[X,Y][s,t]$ belongs to $N\cap J$ if and only if there are
  $\alpha_1,\dots,\alpha_r\in\set K[X,Y]$ and $\beta_1,\dots,\beta_k\in\set K[X,Y][s,t]$
  such that
  \begin{alignat*}1
    q &= \alpha_1n_1 + \cdots + \alpha_rn_r\\
      &= \beta_1p_1 + \cdots + \beta_kp_k.
  \end{alignat*}
  By taking the difference of these two representations of~$q$, we see that the relevant
  tuples
  \[
  (\alpha_1,\dots,\alpha_r,\beta_1,\dots,\beta_k)
  \]
  are precisely the elements of
  \[
    \syz(n_1,\dots,n_r,-p_1,\dots,-p_k)\cap\set K[X,Y]^r\times\set K[X,Y][s,t]^k.
  \]
  We can first compute a Gr\"obner basis of the syzygy module in $\set K[X,Y][s,t]^{r+k}$, then
  discard the lower $k$ coordinates, and then eliminate $s$ and~$t$.
  This yields a basis of the $\set K[X,Y]$-module that contains a tuple
  $(\alpha_1,\dots,\alpha_r)\in\set K[X,Y]^r$ if and only if $\alpha_1n_1+\cdots+\alpha_rn_r\in N\cap J$.
  A basis of this module thus translates into a basis of $N\cap J$.
\end{proof}

We now turn to Step~3 of Alg.~\ref{alg:19}, where we have to compute the intersection of a finitely generated
$\set K[X]$-submodule $M_X$ of $\set K[X,Y][s,t]^2$ with a finitely generated $\set K[Y]$-submodule $M_Y$ of
$\set K[X,Y][s,t]^2$. The result is a $\set K$-vector space, and the task is to compute a basis of this vector
space.

Let $b_1,\dots,b_u$ be a basis of~$M_X$ and $c_1,\dots,c_v$ be a basis of~$M_Y$.
Like in the proof of Lemma~\ref{lemma:cap}, we seek $\alpha_1,\dots,\alpha_u\in\set K[X]$
and $\beta_1,\dots,\beta_v\in\set K[Y]$ such that
\begin{equation}\label{eq:2}
 \alpha_1b_1+\cdots+\alpha_ub_u=\beta_1c_1+\cdots+\beta_vc_v.
\end{equation}
If we can get hold of a finite set of monomials that contains all the monomials which can possibly appear in
$\alpha_1,\dots,\alpha_u,\beta_1,\dots,\beta_v$, then we can find
$\alpha_1,\dots,\alpha_u,\beta_1,\dots,\beta_v$ by making an ansatz with undetermined coefficients,
plugging it into the above equation, comparing coefficients, and solving a linear system over~$\set K$.
Every solution vector translates into a solution $(\alpha_1,\dots,\alpha_u,\beta_1,\dots,\beta_v)\in\set K[X]^u\times\set K[Y]^v$
of equation~\eqref{eq:2}, and every such solution translates into an
element $\alpha_1b_1+\cdots+\alpha_ub_u$ of the intersection $M_X\cap M_Y$.
The following lemma tells us how to find the required monomials.

\begin{lemma}
  Let $(\alpha_1,\dots,\alpha_u,\beta_1,\dots,\beta_v)\in\set K[X]^u\times\set K[Y]^v$ be
  a solution of \eqref{eq:2}, let $i\in\{1,\dots,v\}$, and let $\tau=y_1^{e_1}\cdots y_m^{e_m}$
  be a monomial appearing in~$\beta_i$.
  Let $G$ be a Gr\"obner basis of
  \[
  \syz(b_1,\dots,b_u,-c_1,\dots,-c_v)\subseteq\set K[X,Y]^{u+v}
  \]
  with respect to a TOP order that eliminates~$Y$. Then there exists a monomial $\sigma=x_1^{\varepsilon_1}\cdots x_n^{\varepsilon_n}$
  and an element $g\in G$ such that the first $u$ components are free of $Y$ and the $(u+i)$th component contains
  the monomial $\sigma\tau$.
\end{lemma}
\begin{proof}
  A vector in $\set K[X]^u\times\set K[Y]^v$ is a solution of \eqref{eq:2} if and only if it belongs to the syzygy module.
  The given solution $q$ must therefore reduce to zero modulo~$G$.
  By the choice of the term order, only elements of $G$ whose first $u$ components are free of $Y$ will be used during the reduction.
  Call these elements $g_1,\dots,g_\ell$.
  Again by the choice of the term order, these elements of $G$ will only be multiplied by elements of $\set K[X]$ during
  the reduction, i.e., we will have $q=q_1g_1+\cdots+q_\ell g_\ell$ for certain $q_1,\dots,q_\ell\in\set K[X]$.
  The $(u+i)$th component of $q$ contains the monomial~$\tau$,
  so this monomial appears in a $\set K[X]$-linear combination of the $(u+i)$th components of $g_1,\dots,g_\ell$.
  As $\set K[X]$-linear combinations cannot create new $Y$-monomials, some $\set K[X]$-multiple of $\tau$ must
  already appear in at least one of the $g_1,\dots,g_\ell$.
\end{proof}

With the help of this lemma, we obtain for each $i\in\{1,\dots,v\}$ a finite list of candidates of monomials that
may appear in~$\beta_i$. Applying the lemma again with the roles of $X$ and $Y$ exchanged, we can
also obtain for each $i\in\{1,\dots,u\}$ a finite list of candidates of monomials that may appear in~$\alpha_i$.
This is all we need in order to complete Step~3 of Alg.~\ref{alg:19}.

\begin{example}
  Let us use Alg.~\ref{alg:19} to search for a nontrivial element of $A(I)$ for the ideal
  \[
  I=\<y_1^2-x_2 y_2,x_2^2-x_1 y_1,x_1^4 x_2 y_1-x_2 y_1 y_2^4>.
  \]
  The corresponding ideal $\bar I$ has dimension~$0$, and $A(\bar I)$ contains
  $(s^6,0)$ and $(0,t^6)$. Taking these elements as $(F_1,G_1)$ and $(F_2,G_2)$,
  we find in Step~1 that $M$ is generated by the following vectors:
  \begin{alignat*}1
    &\binom{0}{x_2 y_1^4t^6-x_1 y_1^3 y_2t^6},
   \binom{x_2^4 y_1s^6-x_1 x_2^3 y_2s^6}{0},
   \binom{x_1^3 x_2^3s^6}{y_1^3 y_2^3t^6},\\
   &\binom{x_2^6 y_2^6s^6}{y_1^{12}t^6},
  \binom{x_2^7 y_2^5s^6}{x_1 y_1^{11}t^6},
  \binom{x_2^8 y_2^4s^6}{x_1^2 y_1^{10}t^6},
  \binom{x_2^9 y_2^3s^6}{x_1^3 y_1^9t^6},\\
  &\binom{x_2^{10} y_2^2s^6}{x_1^4 y_1^8t^6},
  \binom{x_2^{11} y_2s^6}{x_1^5 y_1^7t^6},
  \binom{x_2^{12}s^6}{x_1^6 y_1^6t^6},\binom11.
  \end{alignat*}
  In Step~2, we find
  \[
  M_X=\left \< \binom{0}{x_2 y_1^4t^6-x_1 y_1^3 y_2t^6},
  \binom{x_1^3 x_2^3s^6}{y_1^3 y_2^3t^6},
  \binom{x_2^{12}s^6}{x_1^6 y_1^6t^6},\binom11\right>
  \]
  and
  \[
  M_Y= \left \<
  \binom{x_2^4 y_1s^6 -x_1 x_2^3 y_2s^6}{0},
  \binom{x_1^3 x_2^3s^6}{y_1^3 y_2^3t^6},
  \binom{x_2^6 y_2^6t^6}{y_1^{12}t^6},\binom11
\right  >
  \]
  Step~3 yields
  \[
  M_X\cap M_Y=\Span_{\set K}\left(\binom{x_1^3 x_2^3s^6}{y_1^3 y_2^3t^6},\binom11\right),
  \]
  and the final result is $(x_1^3 x_2^3,y_1^3 y_2^3)$.
\end{example}

At the end of the day, Alg.~\ref{alg:19} also has to solve a linear system, but it can be expected that the size of
these linear systems grows more moderately than in the naive approach sketched at the beginning of the section.
On the other hand, Alg.~\ref{alg:19} achieves this size reduction via Gr\"obner basis computations,
so it is not clear which of the two approaches is better.
It is noteworthy however
that the two approaches are not equivalent. For example, if $A(\bar I)$ happens to be trivial, then $A(I)$ is trivial as well,
and therefore detected by the reduction to the bivariate case. The approach based exclusively on linear algebra cannot detect that.

Unlike in the case of principal ideals, it is easy to find examples where $A(I)$ is trivial but $A(\bar I)$ is not.

\begin{example}\label{ex:final}
  Consider the ideal $I\subseteq K[x_1,x_2,y_1,y_2]$ generated by $-x_1 + y_1 + x_1 x_2 y_2 - x_2 y_1 y_2$ and $-x_1 + y_1 + x_1^2 y_1 - x_1 y_1^2$.
  As its generating set is a Gr\"obner basis, it is clear that $I$
  cannot contain any separated polynomials, because in order to reduce a separated
  polynomial to zero, the Gr\"obner basis would need elements with a leading term
  only involving $x_1,x_2$ or only involving $y_1,y_2$.
  On the other hand, for the ideal $\bar I=\<-s x_1 + t y_1 + s^2 t x_1 x_2 y_2 - s t^2 x_2 y_1 y_2, -s x_1 + t y_1 +
 s^2 t x_1^2 y_1 - s t^2 x_1 y_1^2>\subseteq\set K(x_1,x_2,y_1,y_2)[s,t]$
  we have $\bar I=\<sx_1 - ty_1>$ and therefore $A(\bar I)$ is different from $\mathbb{K}((1,1))$.
\end{example}

\section{Conclusion}

We made some progress on the problem of separating variables in multivariate polynomial ideals.
While the algorithm for ideals of dimension zero generalizes smoothly from the bivariate case
to the multivariate case, we did not find a straightforward generalization of the construction
for principal ideals. Instead, we showed that it is possible to reduce the multivariate case
to the bivariate case by merging variables. As a result, we obtain that the algebra of separated
polynomials is simple for every principal ideal generated by a polynomial involving at least one
variable from each of the two groups of variables. It follows furthermore that the algebra is
finitely generated for every ideal that is the intersection of a principal ideal and an ideal
of dimension zero. For arbitrary ideals, however, the algebra may not be finitely generated.
In this case, we can enumerate generators of the algebra, but it remains open whether it is
possible to arrange the enumeration in such a way that it terminates whenever the algebra
happens to be finitely generated. 

\bibliographystyle{plain}
\bibliography{bib}

\end{document}